\documentclass[12pt]{article}

\usepackage{amsmath, amsfonts, amssymb, amsthm}
\usepackage{enumerate, color}
\usepackage{multido}
\usepackage{url}
\usepackage{mathrsfs}
\usepackage{mathtools}
\usepackage{verbatim}
\usepackage{setspace}

\usepackage{soul}
\usepackage{cleveref}
\usepackage[normalem]{ulem}
\usepackage[makeroom]{cancel}

\usepackage[margin=3cm]{geometry}
\usepackage{tikz}

\newif\iffinal
\finaltrue 
\iffinal\else\usepackage[notref,notcite]{showkeys}\fi

\newtheorem{theorem}{Theorem}[section]
\newtheorem{lemma}[theorem]{Lemma}
\newtheorem{remark}[theorem]{Remark}
\newtheorem{proposition}[theorem]{Proposition}

\newtheorem{assumption}[theorem]{Assumption}

\crefname{theorem}{Theorem}{Theorems}
\crefname{lemma}{Lemma}{Lemmas}
\crefname{proposition}{Proposition}{Propositions}
\crefname{section}{Section}{Sections}
\crefname{assumption}{Assumption}{Assumptions}
\crefname{equation}{Eq.}{Eqs.}

\newcommand{\ba}{\begin{align}}
\newcommand{\ea}{\end{align}}

\newcommand{\ket}[1]{|{#1}\rangle}
\newcommand{\bra}[1]{\langle{#1}|}

\renewcommand{\d}[1]{\nabla_{#1}}

\newcommand\otimesal{\mathop{\hbox{\raise 1.6 ex
  \hbox{$\scriptscriptstyle\mathrm{al}$}
\kern -0.92 em \hbox{$\otimes$}}}}
\newcommand\oplusal{\mathop{\hbox{\raise 1.6 ex
  \hbox{$\scriptscriptstyle\mathrm{al}$}
\kern -0.92 em \hbox{$\oplus$}}}}
\newcommand\Gammal{\hbox{\raise 1.7 ex
\hbox{$\scriptscriptstyle\mathrm{al}$}\kern -0.50 em $\Gamma$}}

\newcommand{\caN}{{\mathcal N}}

\newcommand{\scrN}{{\mathscr N}}

\newcommand{\opunit}{\text{1}\kern-0.22em\text{l}}

\newcommand{\e}{{\mathrm e}}

\renewcommand{\d}{{\mathrm d}}

\newcommand{\beq}{ \begin{equation} }
\newcommand{\beqs}{ \begin{equation*} }
\newcommand{\eeqs}{ \end{equation*} }
\newcommand{\eeq}{ \end{equation} }
\newcommand{\bet}{ \begin{theorem} }
\newcommand{\eet}{ \end{theorem} }

\newcommand{\tr }{\mathrm{tr}}

\newcommand{\nk}{\hat{\mathcal{N}}_k}
\newcommand{\n}{\hat{\mathcal{N}}_\infty}
\newcommand{\id}{\mathsf{1}}

\usepackage{authblk}

\begin{document}
\title{Non-demolition measurements of observables with general spectra}
\author[1]{\small M. Ballesteros} 
\author[2]{N. Crawford} 
\author[3]{M. Fraas} 
\author[4]{J. Fr\"ohlich}
\author[5]{B. Schubnel}
\affil[1]{Department of Mathematical Physics, Applied Mathematics and Systems Research Institute (IIMAS),  National Autonomous University of Mexico (UNAM)}
\affil[2]{Department of Mathematics, Technion}
\affil[3]{Instituut voor Theoretische Fysica, KU Leuven  }
\affil[4]{Institut f{\"u}r Theoretische Physik,   ETH Zurich}
\affil[5]{Swiss Federal Railways (SBB)}

\renewcommand\Authands{ and }

\normalsize 
\date \today

\maketitle

\begin{abstract} 
It has recently been established that, in a non-demolition measurement of an observable $\mathcal{N}$ with a finite point spectrum, the density matrix of the system approaches an eigenstate of $\mathcal{N}$, i.e., it ``purifies'' over the spectrum of $\mathcal{N}$. We extend this result to observables with general spectra. It is shown that the spectral density of the state of the system converges to a delta function exponentially fast, in an appropriate sense. Furthermore, for observables with absolutely continuous spectra, we show that the spectral density approaches a Gaussian distribution over the spectrum of $\mathcal{N}$. Our methods highlight the connection between the theory of non-demolition measurements and classical estimation theory.
\end{abstract}

\section{Introduction}
In an indirect measurement, information about a quantum system $S$ is obtained by performing a sequence of standard von~Neumann  measurements on probes that have previously interacted with $S$. A theory of indirect measurements has been proposed by Kraus \cite{Kraus}. Upon tracing out the degrees of freedom of the probes, 
the effective time evolution of the system is described by jump operators, $V_\xi$, indexed by probe measurement outcomes $\xi$, which act on the Hilbert space of pure state vectors of the system $S$. These operators encode the statistics of measurement results and the conditional evolution of the system. If a result $\xi$ is recorded in a probe measurement, the state, 
$\vert \psi \rangle$, of the system changes according to the rule
\begin{equation}
\label{eq:change}
\ket{\psi} \quad  \to \quad  \frac{V_\xi \ket{\psi}}{||V_\xi \ket{\psi}||} .
\end{equation}
In order to describe the probabilities of different measurement outcomes, $\xi$, the set, $\mathcal{X}$, 
of all such outcomes must be equipped with a probability measure. In particular, we need to specify an a-priori measure, $\mu$, ``counting'' the different measurement results. The probability measure on $\mathcal{X}$ is then chosen to be
\begin{equation}
\label{eq:stat}
\bra{\psi} V_\xi^* V_\xi \ket{\psi} d \mu(\xi).
\end{equation}
Consistency imposes a normalisation condition on the jump operators, namely 
$$\int_\mathcal{X} V_\xi^* V_\xi \d \mu(\xi) =1.$$
 Apart from this condition, the operators $V_\xi$ can be chosen arbitrarily.
\mbox{Eqs.~(\ref{eq:change}, \ref{eq:stat})} are consequences of the \textit{Born rule}.  The precise form of the jump operators $V_\xi$ can be derived from the joint Hamiltonian evolution of the system and the probes and the Born rule applied to the probe measurements. We do not repeat this derivation here; but see, e.g., \cite{Holevo}. In our analysis we will make a fixed choice of jump operators. 

If the system interacts with a sequence of independent probes, Eqs.~(\ref{eq:change}, \ref{eq:stat}) can be iterated so as to obtain the probability of recording a sequence, $\underline\xi = (\xi_1,\,\xi_2, \dots)$, of measurement results, along with the corresponding changes of the state of the system. A fundamental problem in the theory of indirect measurements is to understand the asymptotic behaviour of the probability distribution on the space of sequences 
$\underline\xi$ of measurement outcomes and of the state of the system. Various aspects of this problem have been studied by different authors: Conditions for asymptotic purification have been given in \cite{Maassen}, entropy production has been studied in \cite{Benoist2016}, and conditions for uniqueness of the invariant measure have been derived in \cite{BFPP17}; a general approach has been outlined in \cite{BFFS}.

We consider a special case of indirect measurements -- so called non-demolition measurements -- with the feature that all jump operators $V_\xi$ are functions of the observable $\mathcal{N}$ of $S$ that one wants to measure. The motivation to study this case comes from experiments carried out in the group of Haroche and Raimond \cite{guerlin} whose theoretical description fits into the framework developed in our paper. In these experiments it is observed that the state of the system gradually approaches an eigenstate of a certain observable $\mathcal{N}$ (the number of photons stored in a cavity). A theoretical description of this phenomenon has been proposed in \cite{BaBe}, see also \cite{Maassen}, and studied more fully in a series of papers; see \cite{BaBeBe1, BaBeBe2, BFFS}, where further details are provided. In the present paper, we continue this line of research by relaxing the conditions on the spectrum of the observable $\mathcal{N}$; (in all previous works only observables with discrete spectra have been considered).

We will exhibit the phenomenon of ``purification'' over the spectrum of $\mathcal{N}$ for the example of an observable with a general spectrum: Let $\ket{\psi_k}$ denote the state of the system $S$ after the $k^{th}$ probe measurement. We show that, for all continuous functions $f$,  $\bra{\psi_k} f(\mathcal{N}) \ket{\psi_k}$ approaches the value $f(\nu)$, for some point $\nu$ in the spectrum of $\mathcal{N}$, as $k \rightarrow \infty$, and that the frequency of occurrence of a specific point $\nu$ is given by Born's rule applied to the initial state, $\ket{\psi_0}$, of $S$. We also determine the rate of approach to the limit, as $k\rightarrow \infty$. (For precise statements of assumptions and results see Theorem~\ref{thm:1}.) 

If the spectrum of $\mathcal{N}$ is non-degenerate and discrete the convergence of the spectral measure of 
$\ket{\psi_k}$ implies that $\ket{\psi_k}$ approaches an eigenstate  of $\mathcal{N}$, as $k\rightarrow \infty$. If the observable $\mathcal{N}$ has continuous spectrum this conclusion does not hold, because there are no normalisable eigenstates associated with points $\nu$ in the continuous spectrum of $\mathcal{N}$.  If, however, the spectrum of $\mathcal{N}$ is purely absolutely continuous then states of $S$ can be represented as functions, $\psi(\nu)$, on the spectrum of $\mathcal{N}$ that are square-integrable with respect to some measure absolutely continuous with respect to Lebesgue measure. Under suitable hypotheses, the wave function, $\psi_{k}(\nu)$, of the state $\ket{\psi_k}$ then turns out to approach a Gaussian function whose width shrinks to $0$, as $k \rightarrow \infty$. The exact description of this convergence result and the assumptions implying that the states $\ket{\psi_k}$ converge, as $k \rightarrow \infty$, are given in Theorem~\ref{T:CLT2}.

To arrive at this generalisation of known results has required a shift in perspective: It turns out to be useful to map the quantum-mechanical problem (or, at least, parts of it) onto a problem of classical parameter estimation. The phenomenon of purification over the spectrum of an observable then turns into the one of consistency of the maximum likelihood estimator. The hypotheses required for consistency  are well known. For our readers' convenience, we will present detailed proofs using quantum mechanical language. Appropriate references for the results underlying our analysis will be provided as well.

In the next section we describe our setting and summarise our main results. We divide these results into law-of-large-number type results valid for arbitrary spectra of the observable, provided some regularity conditions hold, and central-limit type results for which absolute continuity of the spectrum of the observable is needed. In accordance with this division, proofs are presented in Section~\ref{sec:LLN} and Section~\ref{sec:CLT}, respectively.

\section{Setup and Main Results}

General mixed states of a physical system $S$ are density matrices, $\rho$, acting on a separable Hilbert space 
$\mathcal{H}$. Let $\mu$ be a $\sigma$-finite (counting) measure on a measure space 
$(\mathcal{X}, \Sigma_P)$ of probe measurement outcomes. We consider a family of measurable bounded operators $V_\xi$ acting on $\mathcal{H}$ required to satisfy the normalisation condition\begin{equation}
\label{eq:normalization}
\int_\mathcal{X} V_\xi^* V_\xi \,\d \mu(\xi) = 1.
\end{equation}
Furthermore, we introduce the space, $\Xi \equiv \mathcal{X}^\mathbb{N}$, of infinite sequences $\underline\xi \equiv (\xi_1,\,\xi_2,\dots)$ of probe measurement outcomes equipped with the standard sigma algebra, $\Sigma$, generated by cylinder sets. The measure space $(\Xi,\Sigma)$ has a natural filtration $(\Sigma_k)_{k=1,2,\dots}$, where $\Sigma_k$ consists of sets determined by the first $k$ measurement results. Without danger of confusion we will identify a cylinder set $E \in  \Sigma_k$ with its base $E \in \Sigma_P^{\times k}$. We define an operator-valued stochastic process, $V_k$, adapted to the above filtration by setting
$$
V_k(\underline \xi) = V_{\xi_k} \dots V_{\xi_1}.
$$
With every density matrix $\rho$ on $\mathcal{H}$ and every $k=1,2, \dots$, we associate a probability measure, 
$\mathbb{P}_{\rho}^{(k)}$, on $(\Xi, \Sigma_k)$ by setting
\begin{equation}
\label{eq:2d}
\mathbb{P}_{\rho}^{(k)}(E) := \int_{E} \tr(V_k(\underline{\xi}) \rho V^*_k(\underline{\xi})) \d \mu(\xi_1) \dots \d \mu(\xi_k), \quad E \in \Sigma_k.
\end{equation}
By a well known lemma due to Kolmogorov, these measures determine a unique measure, $\mathbb{P}_{\rho}$, on the space $(\Xi, \Sigma)$.
We define a ``posterior state'' by
\begin{equation}\label{posterior}
\rho_k(\underline \xi):= \frac{V_k(\underline{\xi}) \rho V^*_k(\underline{\xi})}{\tr(V_k(\underline{\xi}) \rho V^*_k(\underline{\xi}))}.
\end{equation}
If the initial state $\rho$ is a rank-one projection, $\ket{\psi_0}\bra{\psi_0}$, then the the posterior state is a rank-one projection, $\ket{\psi_k(\underline{\xi})}\bra{\psi_k(\underline{\xi})}$, with
$
\ket{\psi_k(\underline{\xi})} = V_k(\underline{\xi}) \ket{\psi_0}/||V_k(\underline{\xi}) \ket{\psi_0}||.
$

In the case of non-demolition measurements considered in this paper, the operators $V_\xi$ are functions of a self-adjoint bounded operator $\mathcal{N}$, where $\mathcal{N}$ is the ``observable'' to be measured. By functional calculus, the operators $V_\xi \equiv V_\xi(\mathcal{N})$ are then determined by a measurable family, $V_\xi : \nu \in \mathbb{R}  \mapsto V_\xi(\nu) \in \mathbb{C}$, of bounded complex-valued functions satisfying a normalisation condition
$$
\int_\mathcal{X} |V_\xi(\nu)|^2 \d \mu(\xi) = 1, \quad \mbox{for all} \quad \nu \in \mathbb{R}.
$$
This normalisation condition implies Eq.~(\ref{eq:normalization}). Let $\d \lambda_{\rho}(\nu)$ denote the spectral measure of $\mathcal{N}$ with respect to a state $\rho$, (i.e.,
$ \tr(g(\mathcal{N}) \rho) = \int_{\sigma(\mathcal{N})} g(\nu) \d \lambda_{\rho}(\nu)$, for an arbitrary Borel-measurable, integrable, function $g$). For a cylinder set $E$ of the form $E_1 \times \dots \times E_k$, with  $E_j \in \Sigma_{P}, j=1,\dots,k$, the measure $\mathbb{P}_{\rho}$ introduced in (\ref{eq:2d}) is given by
\begin{equation}
\mathbb P_\rho(E) = \int_{\sigma(\mathcal{N})} \d \lambda_\rho (\nu) \mu_\nu(E_1) \dots \mu_\nu(E_k). \quad 
 \quad \mu_\nu(E_j) = \int_{E_j} \d \mu(\xi) f(\xi|\nu), \label{eq:5}
\end{equation}
where 
\begin{equation}\label{condprob}
f(\xi | \nu) := |V_\xi(\nu)|^2
\end{equation}
acquires the meaning of a conditional probability distribution.
We denote by 
$\mathbb{E}_\nu[\cdot]$ the expectation with respect to the measure $\mu_\nu$. Eq.~(\ref{eq:5}) is the \textit{de Finetti decomposition} \cite{Finetti} of the measure $\mathbb P_\rho$.

Let 
$$
l(\nu | \xi) := \log f(\xi | \nu),
$$
be the log-likelihood function, and define by 
\begin{equation}
\label{eq:likelihood}
l_k(\nu) \equiv l_k(\nu | \underline \xi) := \frac{1}{k} \sum_{j=1}^k l(\nu | \xi_j)
\end{equation}
the log-likelihood function of a sequence of $k$ measurements. The maximum-likelihood estimator of the value of $\mathcal{N}$ is then given by
\begin{equation}\label{maxlikelihood}
\hat{\mathcal{N}}_k := \mathrm{argmax}_{\nu \in \sigma(\mathcal{N})} l_k(\nu).
\end{equation}
For any given realization of $\underline \xi$, there may be more than one $\nu$ for which the RHS is maximized. If such an ambiguity arises we choose the value of $\hat{\caN}_k$ from the set of maximas according to some predetermined rule. In \cite{AliprantisBorder}, Theorem 18.19, it is proved that this can be done in a measurable fashion under the hypotheses adopted in our paper.


Following the notational convention introduced in the previous paragraph, we mostly forgo the $\underline{\xi}$- dependence of functions. We write $\rho_k \equiv \rho_k(\underline{\xi}), V_k \equiv V_k(\underline{\xi})$, etc. The $\underline{\xi}$- dependence is re-introduced at points where we feel that it will be helpful in following our arguments.

\subsection{Law of Large Numbers}

We prove convergence results for $\nk$ and for the states $\rho_k$ introduced in \eqref{posterior} under minimal hypotheses corresponding to assumptions required for the consistency of a maximum likelihood estimator.

Throughout this article, we assume that $\mathcal{N}$ is a bounded operator. Hence $\sigma(\mathcal{N})$ is a compact subset of $\mathbb{R}$ equipped with the induced metric. 

\begin{assumption}
\label{ass:LLN}
For each $\nu \in \sigma(\mathcal{N})$, the function $f(\xi | \nu)$ belongs to $L^1(\mathcal{X}, \d \mu)$ and has the following further properties.
\begin{enumerate}
\item {\bf Identifiability.} The map
$$
\nu \in \sigma(\mathcal{N}) \quad \mapsto \quad f(\xi |\nu) \in  L^1(\mathcal{X}, \d \mu)
$$
is injective, i.e., for $\nu \neq \nu'$, the functions $f(\xi | \nu)$ and $f(\xi | \nu')$ are not identical.
\item {\bf Continuity.}  For every $\xi \in \mathcal{X}$, the function $f(\xi | \nu)$ is continuous on the spectrum
$\sigma(\mathcal{N})$ of the observable $\mathcal{N}$.
\item {\bf Dominance.}  The log-likelihood function $l(\nu | \xi) = \log f(\xi | \nu)$ is dominated in the sense that
$$
\sup_{\nu' \in \sigma(\mathcal{N})} | l(\nu' | \xi)| \in L^1(\mathcal{X}, f(\xi| \nu) \d \mu(\xi)),
$$
for all $\nu \in \sigma(\mathcal{N})$. 
\end{enumerate}
\end{assumption}

The last part of this assumption guarantees that the \textit{relative entropy},
\begin{equation}
\label{eq:8}
S(\nu  | N) = \min_{\nu' \in N}  \mathbb{E}_\nu [l(\nu|\xi) - l(\nu'|\xi)],
\end{equation}
is well defined for any point $\nu$ and any closed subset $N$ of  the spectrum $\sigma(\mathcal{N})$. Jensen's inequality implies that the relative entropy is non-negative, $S(\nu  | N) \geq 0$, and Assumption~\ref{ass:LLN}.1 implies that $S(\nu  | N) = 0$ if and only if $\nu \in N$. Moreover,  from \cite{AliprantisBorder}, Theorem 18.19, we infer that $  S(\cdot  | N) $ is measurable.  

For a Borel set $N\in \sigma(\mathcal{N})$, we let $\Pi(N)$ denote the spectral projection of $\mathcal{N}$ associated with $N$. 

\begin{theorem}
\label{thm:1}
Given Assumption~\ref{ass:LLN}, the maximum likelihood estimator $\hat{\mathcal{N}}_k$ converges almost surely to a random variable $\hat{\mathcal{N}}_\infty$, and, for any Borel set $N \subset \sigma(\mathcal{N})$,
\begin{align}\label{t22.1}
\mathbb P_\rho (\underline \xi : \hat{\mathcal{N}}_\infty \in N) = \tr( \Pi(N) \rho).
\end{align}
Furthermore, if $N$ is the closure of an arbitrary open subset $O$ of the spectrum $\sigma(\mathcal{N})$ of the operator $\mathcal{N}$ contained in the support of the measure $\lambda_\rho$, then 
\begin{align}\label{t22.2}
- \lim_{k \to \infty} \frac{1}{k} \log \tr( \Pi(N) \rho_k) = S(\n | N ), \quad \mathbb P_\rho-\mbox{almost surely}.
\end{align}
\end{theorem}

The first part of the theorem says that the density matrices $\rho_k, k=1,2,\dots,$ purify over the spectrum of $\sigma(\mathcal{N})$. As $k \to \infty$, the spectral measure of $\mathcal{N}$ associated  with $\rho_k$ concentrates on the point $\hat{\mathcal{N}}_\infty$. The speed of convergence of this concentration is described in the second part of the theorem. It is quantified by a large deviation principle with a rate expressed in terms of the relative entropy (\ref{eq:8}). For observables with discrete spectrum, various versions of this statement have been established previously, see \cite{ BaBeBe1, BaBeBe2, BFFS}.

\subsection{Central Limit Theorem}
 To describe the asymptotic behavior of $\nk$ corresponding to the central limit theorem we require additional regularity assumptions. It is convenient to assume that the probability distributions $f(\xi | \nu)$ are defined for every $\nu \in \mathbb{R}$ in such manner that $f(\xi| \cdot)$ is continuous for all $\xi$, and that $f(\xi|\nu)=1$ for all $\nu$ outside a compact interval containing $ \sigma(\mathcal{N})$.  
 \begin{assumption}
\label{ass:CLT}
For all $\nu \in \mathbb{R}$, the following conditions hold:
\begin{enumerate}
\item {\bf Positivity.} For all $\nu \in \sigma(\mathcal{N})$, the probability distribution function  $f(\xi | \nu)$, see Eq. \eqref{condprob}, is strictly positive.
\item {\bf Continuity.}  The probability distribution $f(\xi| \nu)$ is twice continuously differentiable in $\nu$. 

\item {\bf Integrability.} There exists a function $g(\xi) \in L^1(\mathcal{X}, \d \mu_\nu(\xi))$ such that 
$$ 
 |\partial_\nu^{j} f(\xi|\nu')|  < g(\xi), \quad j=1,2, 
$$
for all $\nu' \in \mathbb{R}$, and the following differentiation under the integral sign holds true, 
\begin{equation}
\label{assF}
\int \partial_\nu^j f(\xi |\nu) \d \mu(\xi) = \partial_\nu^j \int f(\xi | \nu) \d \mu(\xi) \,\,(=0), \hspace{1cm}  j \in \{1, 2 \}.
\end{equation}
 The log-likelihood function is square-integrable, 
 $\partial_\nu l(\nu|\xi) \in L^2(\mathcal{X}, \d \mu_\nu(\xi))$, 
and $\mathbb{E}_{\nu}[- \partial_\nu^2 {l}(\nu | \xi))]  $  is strictly positive, for all $\nu \in \sigma(\mathcal{N})$.
\end{enumerate}

\end{assumption}

The quantity $F(\nu) := \mathbb{E}_{\nu}[(\partial_\nu {l}(\nu | \xi))^2)]$ is known as the ``Fisher Information'' of the family of  distributions $f(\xi | \nu)$. Since $f(\xi|\nu)$ is a probability distribution, the right hand side of Eq.~(\ref{assF}) vanishes. This then implies that the Fisher information is also given by  $F(\nu) = -\mathbb{E}_\nu[\partial_\nu^2l(\nu|\xi)]$. The following theorem is a version of the central limit theorem, adapted to our setting. 
\begin{theorem}
\label{T:CLT}
Suppose Assumptions~\ref{ass:LLN} and \ref{ass:CLT} are valid. Then convergence in distribution, 
\begin{equation}\label{noent}
\sqrt{k}[\hat{\mathcal{N}}_k - \hat{\mathcal{N}}_\infty] \quad \stackrel{d}\longrightarrow \quad \scrN(0 ,F^{- 1}(\hat{\mathcal{N}}_\infty)),
\end{equation}
holds, where $\scrN(0 ,\sigma^2)$ is the normal distribution with zero mean and variance $\sigma^2$. 
\end{theorem} 
In a more restricted setting, such a result has been proven in \cite{BFFS}. Note that, in Eq.~\eqref{noent}, the variance of the limit-distribution is itself a random variable. In more precise terms, the following holds: For every $\nu \in \sigma(\mathcal{N}) $, let $ X_\nu $ be a  random variable with distribution  
$  \scrN(0 ,F^{- 1}(\nu)) $. Pick some $a \in \mathbb{R}$, and define
 $A_k := \Big ( \sqrt{k}\big ( \hat{\mathcal{N}}_k - \n \big ) \Big )^{-1} \big ( (- \infty, a] \big ) $. 
Then (see Eq. \eqref{eq:5}) 
\begin{equation}\label{cortr2}
\lim_{k \to \infty } \mathbb P_\rho(A_k) = \lim_{k \to \infty }  \int_{\sigma(\mathcal{N})} \d \lambda_\rho (\nu) \mu_{\nu}^{\otimes \mathbb{N}}(A_k)
 = \int_{\sigma(\mathcal{N})} \d \lambda_\rho (\nu) \mathbb{P}( X_\nu \in 
 (-\infty, a] ).
\end{equation}
 In view of Lemma \ref{Lprinc} below, it is claimed that $\n = \nu  $, $\mu_{\nu}^{\mathbb{N}}$-almost everywhere. Here $   \mathbb{P}( X_\nu \in 
 (-\infty, a] ) $ denotes the probability that $  X_\nu \in 
 (-\infty, a] $.  
 
In order to conclude asymptotic normality of the posterior states  $\rho_k$, we require further assumptions on  $\mathcal{N}$.
 \begin{assumption} 
 \label{ass:Post}
 ~
\begin{enumerate}
\item {\bf Uniform multiplicity.} We suppose that  $\mathcal{N}$ is of uniform multiplicity $n$,
$
\mathcal{H} =  L^2(\mathbb{R},\, \mathbb{C}^n, \d  \Theta )
$
for a regular Borel measure $\Theta$, and 
\begin{equation}
\label{operator}
 (\mathcal{N}f)(\nu)   =  \nu f(\nu). 
\end{equation}  
\item {\bf Absolute continuity.}  $\Theta$ is absolutely continuous with respect to the Lebesgue measure $\lambda$. We denote its Radon-Nikodym derivative by $h$, $$\d \Theta(\nu) = h(\nu) \d \lambda(\nu).$$
\item {\bf Regularity.}  The function $h$ is supported on the spectrum $\sigma(\mathcal{N})$ of $\mathcal{N}$, and it is continuous and strictly positive on $\sigma(\mathcal{N})$. The boundary of $\sigma(\mathcal{N})$ is assumed to have Lebesgue measure $0$.
\end{enumerate}
\end{assumption}
The notion of uniform multiplicity is discussed, e.g., in \cite[Section VII.2]{Reed-Simon-1}. The first part of the third assumption is redundant, as the support of $\Theta$ equals $ \sigma(\mathcal{N}) $. In Remark~\ref{remark} we explain how to weaken this assumption.   

For a Hilbert space $ \mathcal{K} $, we denote by $\mathcal{B}_1(\mathcal{K})$ the space of trace class operators.
Under Assumption~\ref{ass:Post}, an operator $\tau \in \mathcal{B}_1(\mathcal{H})$ is an integral operator with a matrix-valued integral kernel $\tau(\nu,\,\nu')$, where $\nu, \nu' \in \mathbb{R}$:   
$$
( \tau f)(\nu) = \int_{\mathbb{R}}  \tau(\nu, \nu' ) f(\nu ') d \Theta(\nu') .  
$$   
Theorem 2.12 in \cite{Simon_Trace}  states that if the kernel of $\tau$ is continuous then the trace of $\tau$  is given by 
\begin{align}\label{Tr}
\tr (  \tau ) =  \int   \tr_{\mathbb{C}^n} \Big (\tau(\nu, \nu) \Big )   d \Theta(\nu),
\end{align}
where $ \tr(\cdot) $ denotes the trace with respect to $L^2(\mathbb{R}, \mathbb{C}^n,  d \Theta)$  and 
$  \tr_{ \mathbb{C}^n }(\cdot ) $ is the trace on $n\times n$ complex matrices, $\mathbb{M}_{n}(\mathbb{C})$. Eq.~(\ref{Tr}) remains valid also for an integral kernel $ \tau(\nu, \nu') $ that is continuous for $\nu, \nu'$ in a bounded measurable set $N$ and vanishes for $  \nu $  $ \notin N   $  or $\nu' \notin N$.

For a real valued function $\theta(\nu | \xi)$, the transformation 
\begin{equation}
\label{gauge}
V_\xi(\nu) \quad \mapsto \quad \exp(i \theta(\nu | \xi)) V_\xi(\nu)
\end{equation}
transforms the density matrix by a random gauge transformation
$$
\rho_k(\nu,\nu') \quad \mapsto \quad e^{i k \theta_k(\nu)}  \rho_k(\nu,\nu') e^{-i k \theta_k(\nu)}, 
$$
where
$$
\theta_k(\nu) := \frac{1}{k} \sum_{j=1}^k \theta(\nu | \xi_k).
$$
We address the question of convergence of such transformation as $k$ goes to infinity in Remark~\ref{rem:gauge}. In the main text we fix a convenient gauge.
\begin{assumption}
\label{ass:phase} 
The function $V_\xi(\nu)$ is real and positive.
\end{assumption} 
 
The time evolved density matrix then takes the form
\begin{equation}
\label{eq:8prima}
\rho_k(\nu, \nu') = \frac{e^{\frac{1}{2} k l_k(\nu)} \rho(\nu, \, \nu') e^{\frac{1}{2} k l_k(\nu')}}{\int_{\sigma(\mathcal{N})} e^{  k l_k(\nu)} \tr(\rho(\nu,\nu))   \d \Theta(\nu)}.
\end{equation} 
We recall that $l_k(\nu)$ was defined in Eq~(\ref{eq:likelihood}).
The following theorem shows that the density matrix $\rho_k$ is close to a Gaussian state as $k$ tends to infinity. We define the following normalized Gaussian kernel 
$$
\boldsymbol{G}_F(\nu, \nu') = \frac{1}{  \int   \e^{-\frac{F}{2}x^2} \ d x    } \e^{-\frac{F}{4}(\nu^2 + \nu'^2)}.
$$ 
Moreover, for a density matrix $\rho$ with a continuous kernel $ \rho(\nu, \nu') $  we set 
\begin{align}\label{crho}
c_{\rho}(\nu) := \frac{{ \rho}  (\nu,\nu)}{ \tr_{\mathbb{C}^n} { \rho}(\nu,\nu )} .
\end{align}
When the numerator above vanishes, we set $  c_{\rho}$ to be zero. 

For a fixed sequence $\underline \xi$ we set 
$$ \mathcal{B}_1^{(k)} :=  \mathcal{B}_1 \Big (  L^2\big [\mathbb{R},\, \mathbb{C}^n, h(  \hat{\mathcal{N}}_k + \frac{\nu}{\sqrt{k}} \big ) \d \lambda(\nu) \big ]  \Big).  $$
We recall the definition of Fisher information at a point $\nu$, $F(\nu) = \mathbb{E}_{\nu}[(\partial_\nu {l}(\nu | \xi))^2)]$. 
\begin{theorem}
\label{T:CLT2}  
We require Assumptions ~\ref{ass:LLN}, \ref{ass:CLT}, \ref{ass:Post} and \ref{ass:phase}. Moreover, we assume that  that  $\rho(\nu,\,\nu')$, restricted to   $ \sigma(\mathcal{N}) \times \sigma(\mathcal{N})  $, is continuous.
In addition, we require that for every $\nu$   
\begin{align}\label{assumption}
{\lim_{k \to \infty}e^{- k l_k(\nk)} \int_{\mathbb{R}} e^{k l_k (\nk+ \frac{\nu'}{\sqrt{k}})} \d \lambda (\nu')   = \int_{\mathbb{R}} \e^{-\frac{F(  \nu  )}{2}{\nu'}^2} \d \lambda(\nu')}
\end{align} 
holds $\mu_{\nu}^{\otimes \mathbb{N}}$- almost surely.
  Then, 
\begin{align*} \label{Esta}
\lim_{k \to \infty } \Big \|      \frac{1}{\sqrt{k}} { \rho}_k 
  ( \hat{\mathcal{N}}_k + \frac{\nu}{\sqrt{k}} , \hat{\mathcal{N}}_k + \frac{\nu'}{\sqrt{k}})     -      \frac{ c_{\rho}( \n ) }{h(\n)}  \boldsymbol{G}_{F(\n)}(\nu, \nu')   \Big \|_{ \mathcal{B}^{(k)}_1
 } = 0 ,
\end{align*}
almost surely with respect to the measure $\mathbb{P}_{\rho}$. 
\end{theorem}

Notice that 
Theorem \ref{thm:1} implies that $ \n \in \sigma(\mathcal{N}) $ almost surely. For this reason the denominator $h(\n)$  is almost surely strictly positive.    Assumption  \eqref{assumption} is natural in the context of the Bernstein - von Mises Theorem  \cite[Theorem 21]{Ferguson}.

\begin{remark}
\label{rem:gauge}
Suppose that Assumption~\ref{ass:phase} is not satisfied and consider the polar decomposition of $V_\xi(\nu)$:
\begin{equation}
\label{Vdecomposition}
V_\xi(\nu) = \exp(-i \theta(\nu | \xi)) \exp( \frac{1}{2} l(\nu | \xi)),
\end{equation}
where $\theta(\nu\vert\xi)$ is a phase. Assume that the function $\theta(\nu | \xi)$ is twice continuously  differentiable in $\nu$, for almost all $\xi$, and that for every $\nu' \in \mathbb{R}$ there exists a function $g \in L^1(\mathcal{X}, d \mu_{\nu'}(\xi))$ such that $|\partial_\nu^2\theta(\nu | \xi)| < g(\xi)$, for all $\nu \in \mathbb{R}$. Then the conclusion of Theorem~\ref{T:CLT2} would be 
\begin{multline*}
\frac{1}{\sqrt{k}} e^{i \sqrt{k} \partial_\nu \theta_k(\hat{\mathcal{N}}_k) (\nu - \nu')}   { \rho}_k 
  ( \hat{\mathcal{N}}_k + \frac{\nu}{\sqrt{k}} , \hat{\mathcal{N}}_k + \frac{\nu'}{\sqrt{k}}) \\ \quad   \to \quad      \frac{ c_{\rho}( \n ) }{h(\n)}  \boldsymbol{G}_{F(\n)}(\nu, \nu') \e^{-\frac{i}{2} \mathbb{E}_{\n  } [ \partial_{\nu}^2 \theta( \n ) |\xi)]^2(  \nu^2  -\nu'^2  )}
\end{multline*}
in the same topology as specified in the theorem.
\end{remark}

\section{Law of Large Numbers}
\label{sec:LLN}
The purpose of this section is to prove the convergence results for $\nk$ and $\rho_k$, as $k\rightarrow \infty$, formulated in Theorem~\ref{thm:1}. We split the statements of this theorem into two parts. The main tool employed in the proofs of both parts will be the uniform law of large numbers, which we now recall, see \cite[Theorem 16]{Ferguson}.
\begin{theorem}\label{ULLN}
Let $M$ be a compact metric space, and let $X_j(a) \stackrel{d}{=} X(a), a \in M$, be a sequence of i.i.d. random variables.  Suppose that $X(a)$ is almost surely continuous in $a$, and assume that there is a positive random variable $g$ with a finite first moment such that $|X(a)|  < g$ holds almost surely, for all $a \in M$. Then
$$
\sup_{a \in M} |\frac{1}{N} \sum_{j=1}^N X_j(a) - \mathbb{E}[X(a)]| \quad  \stackrel{a.s.}{\longrightarrow} \quad 0,
$$
as $N$ tends to infinity.
\end{theorem}

The next lemma claims a convergence result for $\nk$ with respect to the measure $\mu_\nu^{\otimes \mathbb{N}}$. This is a classical result first proven in \cite{Wald}. We follow a proof given in \cite{Ferguson}.

\begin{lemma} \label{Lprinc}
Suppose that Assumption~\ref{ass:LLN} holds true. Then   
$$\lim_{k \to \infty}  \hat{\mathcal{N}}_k  = \nu,  $$
almost surely with respect to  $ \mu_\nu^{\otimes \mathbb{N}} $.
\end{lemma} 
\begin{proof}
We fix $\nu \in \sigma(\mathcal{N})$ and consider the random variable $l_k(\nu')$ -- see Eq.~(\ref{eq:likelihood}) -- on the measure space $(\Xi, \mu^{\otimes\mathbb{N}}_\nu)$. Given Assumption~\ref{ass:LLN}, all conditions needed to apply the uniform law of large numbers for $l_k(\nu)$ are satisfied, and we have that
$$
\sup_{\nu' \in \sigma(\mathcal{N})}  | l_k(\nu'|\underline{\xi}) - \mathbb{E}_\nu[l(\nu' | \xi)]| \quad \to \quad0\,, \quad \mu^{\otimes\mathbb{N}}_\nu-\mbox{almost surely}.
$$
In particular, for any closed subset $N$ of $\sigma(\mathcal{N})$,
$$S_k(\nu |N) =\min_{\nu' \in N}   (l_k(\nu) - l_k(\nu')) $$ 
converges almost surely to $S(\nu | N)$. \\
Let $U$ be an open neighborhood of $\nu$ and let $U^c$ its complement in $\sigma(\mathcal{N})$. Then, by Assumption~\ref{ass:LLN}.1, $S(\nu| U^c) > 0$, and we conclude that $\mu^{\otimes\mathbb{N}}_\nu$-almost surely there exists $k_0  \equiv k_0(\underline \xi) $ such that $S_k(\nu| U^c) > 0$, for all $k > k_0$. By definition, 
$$0 = S_k(\nu | \sigma(\mathcal{N}))  = S_k(\nu| \hat{\mathcal{N}}_k),$$ 
whence $\hat{\mathcal{N}}_k$ belongs to $U$ almost surely. 
It follows that $\hat{\mathcal{N}}_k$ converges to $\nu$ almost surely.
\end{proof}

The convergence result for $\nk$ (Eq.~(\ref{t22.1}) in Theorem~\ref{thm:1}) is a direct consequence of the lemma.

\begin{proposition}
\label{thm:1a}
Given Assumption~\ref{ass:LLN} the maximum likelihood estimator $\hat{\mathcal{N}}_k$ converges almost surely, as $k\rightarrow \infty$, and, for any Borel set $N \subset \sigma(\mathcal{N})$,
\begin{equation}\label{ss}
\mathbb P_\rho (\underline \xi : \lim_{k \to \infty} \hat{\mathcal{N}}_k \in N) = \tr( \Pi(N) \rho).
\end{equation}
\end{proposition}
\begin{proof}
By $\mathcal{A}$ we denote the set of points $\underline{\xi} $ for which 
$  \lim_{k \to \infty} \hat{\mathcal{N}}_k( \underline{\xi} )  $ exists. The limit is denoted by $ \hat{\mathcal{N}}_\infty( \underline{\xi} )   $. Since 
$ \limsup_{k} \hat{\mathcal{N}}_k  $  and $ \liminf_{k} \hat{\mathcal{N}}_k  $ are measurable functions, and since 
$\mathcal{A}$ is the set of points on which these two functions coincide, the set $\mathcal{A}$ is measurable.  Lemma \ref{Lprinc} 
 implies that
  $\mu_\nu^{\otimes \mathbb{N}}( \mathcal{A}  ) = 1 $,  for every $\nu$.
Almost sure convergence with respect to $\mathbb{P}_{\rho}$, i.e. $ \mathbb{P}_{\rho}(\mathcal{A}) = 1  $, then follows from Eq.~(\ref{eq:5}): 
\begin{align*}
\mathbb P_\rho (  \mathcal{A}) = \int_{\sigma(\mathcal{N})}\mu^{\otimes\mathbb{N}}_\nu( \mathcal{A} ) \d \lambda_\rho(\nu) = \int_{  \sigma(\mathcal{N}) } \d \lambda_\rho(\nu) = 1.
\end{align*}
Next, we prove Eq. \eqref{ss} using Eq.~(\ref{eq:5}): 
\begin{align*}
\mathbb P_\rho \Big (  \n^{-1} ( N ) \Big ) = \int_{\sigma(\mathcal{N})}\mu^{\otimes\mathbb{N}}_\nu  \Big (  \n^{-1} ( N ) \Big ) \d \lambda_\rho(\nu) = \int_N \d \lambda_\rho(\nu) = \tr( \Pi(N) \rho),
\end{align*}
where we use that $\n =  \nu$, $ \mu^{\otimes\mathbb{N}}_\nu $-almost surely; (see Lemma \ref{Lprinc}).
\end{proof}

The second part of Theorem~\ref{thm:1}, concerning the speed of concentration of $\rho_k$ around $\hat{\mathcal{N}}_\infty$, is the content of the following proposition. 

\begin{proposition}
\label{thm:LLN}
We require Assumption ~\ref{ass:LLN}. Let $N$ be the closure of an arbitrary open subset $O$ of the spectrum $\sigma(\mathcal{N})$ of the operator $\mathcal{N}$.  Suppose that $N$ is contained in the support of the measure $\lambda_\rho$. Then
$$
 -\lim_{k \to \infty} \frac{1}{k} \log \tr( \Pi(N) \rho_k) = S(  \hat{\mathcal{N}}_\infty   | N),  \quad \mathbb{P}_\rho-\mbox{almost surely}.
$$
\end{proposition} 
\begin{proof}
We prove the Proposition in two steps. \\
{\bf Step 1:} We prove that, for every $\nu'$  in the support of the measure $\lambda_\rho$,
\begin{align}\label{lab}
 -\lim_{k \to \infty} \frac{1}{k} \log \tr( \Pi(N) \rho_k) = S( \nu' | N), \quad \mu^{\otimes\mathbb{N}}_{\nu'}-\mbox{almost surely}.
\end{align}
The quantity $\tr( \Pi(N) \rho_k)$ can be expressed in terms of the likelihood function as
$$
\tr(\Pi(N) \rho_k) = \frac{\int_N \d \lambda_\rho(\nu) \e^{k l_k(\nu)}}{\int_{\sigma(\mathcal{N})} \d \lambda_\rho(\nu) \e^{k l_k(\nu)}}.
$$
Let $I := \max_{ \nu \in N} \mathbb{E}_{\nu'}[l(\nu | \xi)]$. We note that Eq. \eqref{lab} follows from
$$
 \lim_{k \to \infty} \frac{1}{k} \log \int_N \d \lambda_\rho(\nu) \e^{k l_k(\nu)} = I, \quad \mu^{\otimes\mathbb{N}}_{\nu'}-\mbox{almost surely},
$$ 
since  $  \mathbb{E}_{\nu'}[l(\nu' | \xi)] = \max_{ \nu \in \sigma(\mathcal{N})} \mathbb{E}_{\nu'}[l(\nu | \xi)]$.   
We fix some $\varepsilon > 0$ and show that, for large enough $k$, 
$$
\Big| \frac{1}{k} \log \int_N \d \lambda_\rho(\nu) \e^{k l_k(\nu)}e^{-Ik} \Big | \leq \varepsilon,
$$
$\mu^{\otimes\mathbb{N}}_{\nu'}$-almost surely.\\
Bound from above: Since $l_k(\nu | \underline{\xi})$ converges uniformly to $\mathbb{E}_{\nu'   }[l(\nu | \xi)]$, by the uniform law of large numbers, we have that $\mu^{\otimes\mathbb{N}}_{\nu'}$-almost surely  
 $l_k(\nu | \underline{\xi}) -I \leq \varepsilon$, for large $k \equiv  k(\underline \xi)$ and all $\nu \in N$. Then
$$
\frac{1}{k} \log \int_N \d \lambda_\rho(\nu) \e^{k l_k(\nu| \underline{\xi})}e^{-I k} \leq \frac{1}{k} \log \int_N \d \lambda_\rho(\nu) \e^{k \varepsilon} \leq \varepsilon.
$$
Bound from below: Consider the open set  
$$U(\varepsilon) := \{\nu \in \sigma(\mathcal{N}) : I - \mathbb{E}_{\nu' }[l(\nu | \xi)]  < \varepsilon \}.$$
Since $I$ is the maximum over $\nu \in N$ of \,$\mathbb{E}_{\nu'}[l(\nu | \xi)]$, we have that \mbox{$ N \cap U(\varepsilon) \ne \emptyset $}. Since $N = \overline{O}$, the interior of  $ N \cap U(\varepsilon) $ is non-empty. Thus, since $ N $ is contained in the support of $ \lambda_{\rho} $, it follows that $ \lambda_{\rho}( N \cap U(\varepsilon)  ) > 0 $, for every $\varepsilon >0$. By the uniform law of large numbers, 
$  \big | \mathbb{E}_{\nu' }[l(\nu | \xi)] - l_k(\nu | \underline{\xi}) \big |  < \varepsilon/4$, for all $\nu$ and for sufficiently large $k \equiv  k(\underline \xi)$, $\mu^{\otimes\mathbb{N}}_{\nu'} $-almost surely.  It follows that
$I - l_k(\nu | \underline{\xi}) < \varepsilon / 2  $, 
$\mu^{\otimes\mathbb{N}}_{\nu'} $-almost surely,   for every $\nu \in U(\varepsilon/4)$ and for sufficiently large $k \equiv  k(\underline \xi)$. Hence 
\begin{align*}
\frac{1}{k} \log \int_N \d \lambda_\rho(\nu) \e^{k l_k(\nu| \underline{\xi})}e^{-I k} & \geq \frac{1}{k} \log \int_{N \cap U(\varepsilon/4)} \d \lambda_\rho(\nu) \e^{-k \varepsilon /2 } \\ & = -\varepsilon/2 + \frac{1}{k} \log \lambda_{\rho}(U(\varepsilon/4))  > - \varepsilon,
\end{align*}
if we choose $k >  2 \varepsilon^{-1} |\log \lambda_{\rho}(U(\varepsilon/4))|$; (note that $N \cap U(\varepsilon/2)$ has a strictly positive measure). \\
{\bf Step 2: } We now prove the proposition.
We set 
\begin{align}
 \mathcal{B} : = \Big \{  \underline \xi \Big |
 -\lim_{k \to \infty} \frac{1}{k} \log \tr( \Pi(N) \rho_k ( \underline \xi )) = S(  \hat{\mathcal{N}}_\infty( \underline \xi )   | N)    \Big  \}.  
\end{align}  
Because the set where $ \lim_{k \to \infty} \frac{1}{k} \log \tr( \Pi(N) \rho_k ( \underline \xi )) $ exists  is measurable, and the function
$S(  \hat{\mathcal{N}}_\infty( \underline \xi )   | N)$ is measurable, the set $ \mathcal{B} $ is measurable. 
Moreover,  Lemma~\ref{Lprinc} implies that   $ \n = \nu' $, $\mu^{\otimes\mathbb{N}}_{\nu'} $-almost surely, which when combined with Step 1 proves that   
\begin{align}\label{labprima}
 \lim_{k \to \infty} -\frac{1}{k} \log \tr( \Pi(N) \rho_k) = 
 S( \nu' | N) = S( \n | N), \quad \mu^{\otimes\mathbb{N}}_{\nu'}-\mbox{almost surely}.
\end{align}
We conclude that $  \mu^{\otimes\mathbb{N}}_{\nu'}(\mathcal{B}) = 1  $, for every 
$ \nu' $ in the support of $\lambda_{\rho}$. Finally, Eq \eqref{eq:5} implies the desired result, $\mathbb P_\rho (  \mathcal{B})=1$.
\end{proof}

\section{Central Limit Theorem}
\label{sec:CLT}
In this section we prove Theorems~\ref{T:CLT} and \ref{T:CLT2}  and conclude by discussing some extensions of our results.  
\subsection{Proof of Theorem~\ref{T:CLT}}
The proof is an adaptation, to our setting, of the proof of a theorem  due  to Cram\'{e}r: see  \cite[Theorem~18]{Ferguson}. 
\begin{proof}[Proof of Theorem~\ref{T:CLT}]
We first fix $\upsilon \in \sigma(\mathcal{N})$.
To understand the relevance of the Fisher Information, one notes that, since $\partial_\nu l_k(\hat{\mathcal{N}}_k)=0$, the mean value theorem implies 
\[
\partial_\nu l_k(\upsilon )= \partial_\nu {l}_k(\upsilon )-\partial_\nu {l}_k(\hat{\mathcal{N}}_k)=-\partial_\nu^2 {l}_k(\nu')[\hat{\mathcal{N}}_k-\upsilon ],
\]  
for some $\nu'$ in the interval between $\upsilon $ and $\nk$.  As a consequence of the central limit theorem, the assumption that the Fisher information is finite implies that $\sqrt{k} \partial_\nu {l}_k(\upsilon )$ converges in distribution (with respect to $\mu_{\upsilon}^{\otimes \mathbb{N}}$)
to a Gaussian random variable with  mean $0$ and variance $ \mathbb{E}_{\upsilon}[(\partial_\nu {l}(\upsilon | \xi))^2)] $; (recall that  $\mathbb{E}_{\upsilon}( \partial_\nu l(\upsilon | \xi)) = 0  $).  

By the uniform law of large numbers
and the fact that $\nk$ converges to $\upsilon$ 
(almost surely with respect to  $\mu_{\upsilon}^{\otimes \mathbb{N}}$), $\partial_\nu^2 {l}_k(\nu')$ 
converges to $\mathbb{E}_{\upsilon}[\partial_\nu^2 {l} (\upsilon | \xi)]$. Furthermore, one may check, using Assumption~\ref{ass:CLT}.2, that, for any point 
$\nu$, $\mathbb{E}_{\nu}[\partial_\nu^2 {l}(\nu |\xi)]=-\mathbb{E}_{\nu}[(\partial_\nu {l}(\nu | \xi))^2)]$.  Combining this with the last statement of the previous paragraph we arrive at the following convergence result:
\begin{equation}\label{corrr1}
\sqrt{k}[\hat{\mathcal{N}}_k - \upsilon] \quad \stackrel{d}\longrightarrow \quad \scrN(0 ,F^{- 1}(\upsilon)),
\end{equation}
where the above convergence is in distribution with respect to the measure $\mu_{\upsilon}^{\otimes \mathbb{N}}$. Take $a \in \mathbb{R} $ and set $A_k = \Big ( \sqrt{k}\big ( \hat{\mathcal{N}}_k - \n \big ) \Big )^{-1} \big ( (- \infty, a] \big )$ , and
$$
 g_k (\upsilon) :=   \mu_{\upsilon}^{\otimes \mathbb{N}}\Big (  \Big (  \sqrt{k}\big ( \hat{\mathcal{N}}_k - \upsilon \big ) \Big )^{-1} \big ( (- \infty, a] \big ) \Big ) . $$  Eq. \eqref{eq:5} implies that 
\begin{equation}\label{corrr2}
\mathbb P_\rho(A_k) = \int_{\sigma(\mathcal{N})} \d \lambda_\rho (\upsilon) \mu_{\upsilon}^{\otimes \mathbb{N}}(A_k) = \int_{\sigma(\mathcal{N})} \d \lambda_\rho (\upsilon) g_k (\upsilon)  , 
\end{equation}
where we used Lemma \ref{Lprinc}. Finally, the desired result follows from the Lebesgue dominated convergence theorem and \eqref{corrr1}, \eqref{corrr2}.

%
\end{proof}

\subsection{Proof of Theorem~\ref{T:CLT2}}\label{SecP}

We consider the isometry
\begin{equation} \label{iotatt}
\iota :  L^2(\mathbb{R},\, \mathbb{C}^n,  d \Theta)    \to   L^2(\mathbb{R},\, \mathbb{C}^n, d \lambda ),  \hspace{.5cm} \phi \mapsto \iota(\phi):=  \sqrt{h} \phi,
\end{equation}
and denote by $\iota^{-1}$ the inverse on its range extended by zero on the orthogonal subspace.
It follows that $\iota$  naturally defines  an isometry  between $ \mathcal{B}_1( L^2(\mathbb{R},\, \mathbb{C}^n,  d \Theta)  ) $ and 
$  \mathcal{B}_1( L^2(\mathbb{R},\, \mathbb{C}^n,  d \lambda )  )   $
given by 
\begin{align}\label{isoiota}
  \mathcal{B}_1( L^2(\mathbb{R},\, \mathbb{C}^n,  d \Theta)  )   \ni
     \tau  \mapsto {\boldsymbol \tau} & :=  \iota \tau \iota^{-1}
    \\ \notag & =  \sqrt{h(\nu)} \tau(\nu, \nu')\sqrt{h(\nu')}  \in 
     \mathcal{B}_1( L^2(\mathbb{R},\, \mathbb{C}^n,  d \lambda)  ) . 
\end{align}

We repeatedly use the above transformation to change $\rho_k$ into ${\boldsymbol \rho}_k$.

\begin{proposition}
\label{P:CLT2}
We require Assumptions ~\ref{ass:LLN}, \ref{ass:CLT}, \ref{ass:Post} and \ref{ass:phase}, and we assume that  
$\rho(\nu,\,\nu')$ is continuous in $(\nu,\nu') \in \sigma(\mathcal{N}) \times \sigma(\mathcal{N})$.
Suppose, moreover, that, for every $\nu$, 
\begin{align}\label{assumptionp}
\lim_{k \to \infty}e^{- k l_k(\nk)} \int_{\mathbb{R}} e^{k l_k (\nk+ \frac{\nu'}{\sqrt{k}})} \d \lambda(\nu')  = \int_{\mathbb{R}} \e^{-\frac{F(  \nu  )}{2}{\nu'}^2} \d  \lambda(\nu')
\end{align}
holds $\mu_\nu$- almost surely.
  Then, in the standard topology of trace class operators on the Hilbert space 
  $L^2(\mathbb{R},\, \mathbb{C}^n, d\lambda )$,
\begin{align} \label{Estaprima}
\lim_{k \to \infty }\frac{1}{\sqrt{k}} {\boldsymbol \rho}_k 
  ( \hat{\mathcal{N}}_k + \frac{\nu}{\sqrt{k}} , \hat{\mathcal{N}}_k + \frac{\nu'}{\sqrt{k}})   =  c_{\rho}(\n) \boldsymbol{G}_{F(\n)}(\nu, \nu')  ,
\end{align}
almost surely with respect to the measure $\mathbb{P}_{\rho}$ .
\end{proposition}

\begin{proof}

Our proof is divided into two steps. \\
{\bf Step 1:} Let $\upsilon $  belong to the interior of  $ \sigma(\mathcal{N})$. It follows that ${\boldsymbol \rho}(\nu, \nu')$ is continuous at the point $(\upsilon, \upsilon)$, viewed as a function in $\mathbb{R}^2$.  We prove that, almost surely with respect to $ \mu_{\upsilon }^{\otimes \mathbb{N}} $,
\begin{align} \label{Estass} 
\lim_{k \to \infty }\frac{1}{\sqrt{k}}  
    {\boldsymbol \rho}_k 
  \Big ( \nk + \frac{\nu}{\sqrt{k}} , \nk + \frac{\nu'}{\sqrt{k}} \Big )
   =  c_{\rho}(\upsilon) \boldsymbol{G}_{F(\upsilon)}(\nu, \nu')  
\end{align} 
in the standard topology of trace class operators on the Hilbert space $L^2(\mathbb{R},\, \mathbb{C}^n,  \d \lambda)$. 
Convergence with respect to the trace norm is, in general, not obvious because there are not many explicit formulas to compute this norm (unlike what happens with the Hilbert-Schmidt norm).  To prove \eqref{Estass} we will use Theorem 2.19 in \cite{Simon_Trace} that asserts that a sequence $ (\tau_k)_{k \in \mathbb{N}} $ converges to $\tau$ with respect to the trace-norm if $\tau_k$ and $\tau_k^*$ converge strongly to $\tau$ and $\tau^*$ (respectively) and the trace-norm of $\tau_k$ converges to the trace-norm of $  \tau $.  Notice that \eqref{Tr} alone does not help because even if $\tau_k$ and $\tau$ are positive, for every $k$, $\tau_k - \tau$ is not necessarily positive and, therefore, formula \eqref{Tr} cannot be used to estimate its trace-norm.

In the remaining of this proof we will use a couple of times the following result of measure theory that is a direct consequence of Theorem 1.21 and Lemma 1.32 in \cite{Kallenberg}: Suppose that $f_k$ are positive functions on $L^p$, $  1 \leq p < \infty $, that converge point-wise (a.e.) to a positive function $f \in L^p $. Moreover,  suppose that $(\eta_k)_{k \in \mathbb{N}}$ is a sequence of bounded measurable functions that is uniformly bounded (i.e., $ \sup_{x,k} |\eta_k(x)| < \infty  $ ) and converges point-wise to a function $ \eta  $ (a.e.). Assume  in addition that  $\int f_k^p \to \int f^p$. Then  
\begin{equation}
\label{obs}
 \lim_{k \to \infty }f_k \eta_k = f \eta, \hspace{1cm} \text{with respect to the $L^p$-norm}. 
\end{equation} 
We use this observation together with Eq.~(\ref{Tr}) to prove convergence with respect to the trace norm. 
  
Set 
\begin{align*} \tau_k  :=  & \frac{1}{\sqrt{k}}   {\boldsymbol \rho}_k 
  \Big ( \nk + \frac{\nu}{\sqrt{k}} , \nk + \frac{\nu'}{\sqrt{k}} \Big ). 
\end{align*} 
    We define 
$$C_k :=   \Big (  \int e^{ k l_k(\nu)}     \tr(\rho(\nu,\nu))  h(\nu) \d \lambda(\nu) \Big ) \e^{-k l_k(\nk)}
  \sqrt{k}   $$ and 
\begin{equation} \label{cor0}
\widetilde{\tau}_k : = C_k \tau_k.
\end{equation}
 Now, we will compute the limit, when $k$ tends to infinity, of $C_k$. Since the trace of $ \tau_k$ equals 1, we have that
 \begin{equation} \label{cor1}
 \lim_{k \to \infty}C_k = \lim_{k \to \infty} \tr( \widetilde{\tau}_k ),
\end{equation}     
whenever one of these limits exists.  Eq. \eqref{eq:8prima} implies that  
\begin{equation}\label{Estas1}
\widetilde{\tau}_k(\nu,\,\nu' )  =    \e^{-k l_k(\nk)}  \e^{\frac{1}{2} k l_k(\nk + \frac{\nu}{\sqrt{k}})} {\boldsymbol \rho}(\nk + \frac{\nu}{\sqrt{k}}, \nk + \frac{\nu'}{\sqrt{k}})     \e^{\frac{1}{2} k l_k(\nk + \frac{\nu'}{\sqrt{k}})}. 
\end{equation} 
 We define the random function    
$$
X_k(\nu ) := \e^{-\frac{1}{2}k l_k(\nk)}  \e^{\frac{1}{2} k l_k(\nk + \frac{\nu}{\sqrt{k}})}.   
$$

By  Taylor's formula -- note that $\partial_{\nu} l_k(\nk) = 0 $, because $\nk$  is where the maximum occurs -- we have that 
$$
l_k(\nk  + \frac{\nu}{\sqrt{k}}) - l_k(\nk)     = \frac{\nu^2}{2k} \partial_\nu^2 l_k(\nu_1)
$$
for some $\nu_1 \in (\nk,\,\nk + \frac{\nu}{\sqrt{k}})$. Hence the random function $X_k$ has the form
$$
X_k(\nu ) = e^{\frac{\nu^2}{4} \partial_\nu^2 l_k(\nu_1)}.
$$
By the uniform law of large numbers (cf. the proof of Theorem~\ref{T:CLT})  $-\partial_\nu^2 l_k(\nu)$ converges uniformly to $  \mathbb{E}_{\upsilon} \big (-\partial_\nu^2 l_k(\nu) \big )  $, which is strictly positive at $\nu = \upsilon$ by Item 3. in Assumption \ref{ass:CLT}.
We define 
$$
\Gamma(\nu) := e^{- \frac{\nu^2}{4} F(\upsilon)}.
$$
 Using Lemma \ref{Lprinc} and the uniform law of large numbers  we have that (almost surely with respect to $\mu_\upsilon^{\otimes \mathbb{N}} $)
\begin{align} \label{cor3}
 X_k(\nu) \to   \Gamma(\nu),   \hspace{.3cm}    {\boldsymbol \rho}(\nk + \frac{\nu}{\sqrt{k}}, \nk + \frac{\nu}{\sqrt{k}})  \to   {\boldsymbol \rho}(\upsilon, \upsilon ),  \hspace{.3cm}
\widetilde{\tau}_k(\nu, \nu')  \to    \Gamma(\nu) {\boldsymbol \rho}(\upsilon, \upsilon)
\Gamma(\nu'),
\end{align}
 almost surely with respect to $\mu_\upsilon^{\otimes \mathbb{N}} $.
 Next we use \eqref{obs} with $f_k(\nu) = X_k(\nu)^2 $, $ f = 
 \Gamma(\nu)^2   $, $\eta_k(\nu) =  \tr_{\mathbb{C}^n} {\boldsymbol \rho}(\nk + \frac{\nu}{\sqrt{k}}, \nk + \frac{\nu}{\sqrt{k}})    $ and $ \eta(\nu)= \tr_{\mathbb{C}^n} {\boldsymbol \rho}(\upsilon, \upsilon ) $ (Eq.~\eqref{cor3} together with Assumption \eqref{assumptionp} are the requirements for \eqref{obs}). We obtain (see Eq. \eqref{Tr}): 
 \begin{align} \label{core}
\lim_{k \to \infty } \tr (\widetilde{  \tau}_k ) = &     \lim_{k \to \infty } \int X_k(\nu)^2  \tr_{\mathbb{C}^n} {\boldsymbol \rho}(\nk + \frac{\nu}{\sqrt{k}}, \nk + \frac{\nu}{\sqrt{k}})  d  \lambda(\nu)   \\ \notag  = &  \int \Gamma(\nu)^2 \tr_{\mathbb{C}^n} {\boldsymbol \rho}(\upsilon, \upsilon ) d  \lambda(\nu) = \tr (\Gamma(\nu)  {\boldsymbol \rho}(\upsilon, \upsilon ) \Gamma(\nu')  ). 
 \end{align}
   Then Eq. \eqref{cor1} imply that 
 (almost surely with respect to $\mu_\upsilon^{\otimes \mathbb{N}} $)
  \begin{equation} \label{cor4}
 \lim_{k \to \infty}C_k    =  \tr(\Gamma(\nu) {\boldsymbol \rho} (\upsilon, \upsilon)
\Gamma(\nu')  )= \int  \Gamma(\nu) \tr_{\mathbb{C}^n} {\boldsymbol \rho}(\upsilon, \upsilon)
\Gamma(\nu)  d \lambda(\nu) . 
\end{equation}
Next we will prove that 
\begin{align}\label{cor5}
s-\lim_{k \to \infty}\widetilde{\tau}_k  = \Gamma(\nu){\boldsymbol \rho} (\upsilon, \upsilon)
 \Gamma(\nu')  
\end{align}
(here $s-\lim$ represents the strong limit), which together with \eqref{cor4} and Theorem 2.19 in \cite{Simon_Trace} implies that the limit in \eqref{cor5} holds true also with respect to the trace norm.  We will actually prove a stronger result, namely that the limit in \eqref{cor5} is valid with respect to the Hilbert-Schmidt
norm. We recall that for a Hilbert-Schmidt operator $O \equiv O(\nu, \nu')$ acting on 
$ L^2(\mathbb{R}, \mathbb{C}^n, d\lambda) $, its  Hilbert-Schmidt norm is given by
\begin{align}
\| O  \|_{{\rm HS}}^2 = \int  \| O(\nu, \nu') \|_{{\rm HS}(\mathbb{C}^n)}^2 d   \lambda(\nu) d  \lambda(\nu') ,
\end{align}
 where $\| \cdot \|_{{\rm HS}(\mathcal{C}^n)}^2 $ is the Hilbert-Schmidt norm for operators in $\mathbb{C}^n$.  
 We use again \eqref{obs}, but now we take $p = 2$.  Moreover, we set $f_k(\nu, \nu') = X_k(\nu) X_k(\nu') $, $ f(\nu, \nu' ) = 
   \Gamma(\nu)  \Gamma(\nu')  $, $\eta_k(\nu, \nu') = {\boldsymbol \rho}(\nk + \frac{\nu}{\sqrt{k}}, \nk + \frac{\nu'}{\sqrt{k}})$ and $ \eta(\nu, \nu')= {\boldsymbol \rho}(\upsilon, \upsilon ) $ (Eq.~\eqref{cor3} together with Assumption \eqref{assumptionp} are the requirements for \eqref{obs}). We obtain: 
\begin{align}\label{cor6}
\lim_{k \to \infty}  \| \widetilde{\tau_k} - &  \Gamma(\nu) {\boldsymbol \rho}(\upsilon, \upsilon) 
\Gamma(\nu')    \|^2_{\rm Hs} \\ \notag & = \lim_{k \to \infty} \int \Big \|    \widetilde{\tau_k}(\nu, \nu') -   \Gamma(\nu) {\boldsymbol \rho}(\upsilon, \upsilon)
  \Gamma(\nu')   \Big \|^2_{\rm Hs (\mathbb{C}^n)} d  \lambda(\nu)  d    \lambda(\nu')  =0.
\end{align}
Eqs. \eqref{core} and \eqref{cor6} then lead to
\begin{equation}\label{cor7}
 \lim_{k \to \infty}\widetilde{\tau}_k(\nu, \nu')  = \Gamma(\nu)
{\boldsymbol \rho}(\upsilon, \upsilon)
\Gamma(\nu'),  
\end{equation}
with respect to  the trace norm (see   Theorem 2.19 in \cite{Simon_Trace} and \eqref{cor4}, \eqref{cor6}). Then, \eqref{cor0}, \eqref{cor4} and \eqref{cor7} imply that 
\begin{align}
 \lim_{k \to \infty}{\tau}_k(\nu,\nu')  = \frac{1}{  \tr \Big (\Gamma(\nu){\boldsymbol \rho}(\upsilon, \upsilon)
\Gamma(\nu')  \Big )} \Gamma(\nu) {\boldsymbol \rho}(\upsilon, \upsilon)
\Gamma(\nu'), 
\end{align} 
which directly implies Eq. \eqref{Estass}. 

{\bf Step 2: } We prove Eq. \eqref{Estaprima}. Set $\mathcal{C}$ be the set of points $ \underline{\xi} $ such that \eqref{Estaprima} holds. As we have argued above (see the proof of Proposition \ref{thm:1a}), the set $\mathcal{C}$ is measurable. Since $\n = \upsilon$, almost surely with respect to $ \mu_{\upsilon }^{\otimes \mathbb{N}} $, (see Lemma \ref{Lprinc}), Step 1 implies that $  \mu_{\upsilon }^{\otimes \mathbb{N}}(\mathcal{C}) = 1 $, for every $\upsilon$ in the interior of $ \sigma(\mathcal{N}) $, that by assumption has $\lambda_\rho$ measure $1$. Finally, Eq \eqref{eq:5} implies the desired result, $\mathbb P_\rho (  \mathcal{C}) =1$. 
\end{proof}

We are ready to prove Theorem~\ref{T:CLT2} 

\begin{proof}[Proof of Theorem~\ref{T:CLT2}] 
We only prove that for every $\upsilon$ in the interior of 
$  \sigma(\mathcal{N}) $ (we abbreviate $F \equiv F(\upsilon)$)
\begin{align} \label{Estaup}
\lim_{k \to \infty } \Big \|     \frac{1}{\sqrt{k}} { \rho}_k 
  ( \hat{\mathcal{N}}_k + \frac{\nu}{\sqrt{k}} , \hat{\mathcal{N}}_k + \frac{\nu'}{\sqrt{k}})   -   \frac{ c_{\rho}( \upsilon ) }{h(\upsilon)}  \boldsymbol{G}_{F}(\nu, \nu')   \Big \|_{ \mathcal{B}^{(k)}_1
 } = 0 ,
\end{align}
almost surely with respect to $\mu_{\upsilon}^{\otimes \mathbb{N}}$. The rest of the
proof follows as in the proof of Step 2 in Proposition \ref{P:CLT2}.  

For every trace class  operator $\tau \equiv \tau(\nu, \nu')$ acting on $ L^2( \mathbb{R}, \mathbb{C}^n, h d\lambda )  $, we set 
$$ \tau^{(k)} (\nu, \nu') = 
\tau( \nk + \nu / \sqrt{k},   \nk + \nu' / \sqrt{k} ), $$
 acting on  $ L^2( \mathbb{R}, \mathbb{C}^n, h( \nk + \nu /\sqrt{k} ) d \lambda )  $. Next we set 
$$  \boldsymbol{ \tau}^{(k)}(\nu, \nu' )  =
 \sqrt{h^{(k)}(\nu)} \tau^{(k)}(\nu, \nu') \sqrt{ h^{(k)}(\nu')} ,  $$ as an operator  in $L^2(\mathbb{R}, \mathbb{C}^n,d\lambda )$.
As we argued in Eq.~(\ref{isoiota}), $\boldsymbol{\tau}^{(k)} $ and $\tau^{(k)}$ have the same norm, in their respective spaces. In the proof of Proposition \ref{P:CLT2} we prove that 
\begin{multline}\label{gas}
\lim_{k  \to \infty }  \Big \| \frac{1}{\sqrt{k}}   \boldsymbol{ \rho }_k^{(k)}  -  c_{\rho}(\upsilon) \boldsymbol{ G}_{F}(\nu, \nu')   \Big  \|_{\mathcal{B}_1 (L^2(\mathbb{R}, \mathbb{C}^n, d \lambda ))} \\ = \lim_{k  \to \infty }    \Big \|      \frac{1}{\sqrt{k}}   \rho _k^{(k)}  - c_{\rho}(\upsilon)  G_{F}(\nu, \nu')   \Big  \|_{\mathcal{B}^{(k)}_1 } = 0,
\end{multline} 
 almost surely with respect to $\mu_{\upsilon}^{\otimes \mathbb{N}}$,  where $G_F(\nu, \nu') := \frac{\boldsymbol{G}_F(\nu, \nu')  }{\sqrt{h^{(k)}}(\nu)  \sqrt{h^{(k)}}(\nu')  }$.
Moreover,
\begin{align}\label{gas1}
 \lim_{k  \to \infty }  &\Big \|    \frac{ c_{\rho}(\upsilon )}{h(\upsilon)}  \boldsymbol{G}_F(\nu,\nu')  -  c_{\rho}(\upsilon)  G_F(\nu,\nu')   \Big  \|_{\mathcal{B}^{(k)}_1 }  \\ & =  \notag  \lim_{k  \to \infty }  \Big \|  \frac{ c_{\rho}(\upsilon)   \sqrt{h^{(k)}(\nu)}  \sqrt{h^{(k)}(\nu')}   }{h(\upsilon )}  \boldsymbol{G}_F(\nu,\nu')  - c_{\rho}(\upsilon) \boldsymbol{  G}_F(\nu,\nu')   \Big  \|_{_{\mathcal{B}_1 (L^2(\mathbb{R}, \mathbb{C}^n, d \lambda ))} } =0.
\end{align}  
The proof of the last statement, which is left to the reader, can be made either by a direct computation or by the same procedure that we applied repeatedly to prove a trace convergence of operators.
Finally \eqref{gas} and \eqref{gas1} imply Eq. \eqref{Estaup}.  
\end{proof}

\begin{remark}
\label{remark}
A) Assumptions~\ref{ass:LLN}, \ref{ass:CLT} require the validity of various conditions for all points in the spectrum $\sigma(\mathcal{N})$ of $\mathcal{N}$. If these conditions only hold true in some open interval $N \subset \sigma(\mathcal{N})$ then our conclusions hold when conditioned on $\n \in N$.

B) By a direct integral version of the spectral decomposition of $\mathcal{N}$, there exists a Hilbert space bundle $\mathcal{H}_\nu$ over a measure space $(\sigma(\mathcal{N}), \lambda)$ such that
\begin{equation}
\label{eq:DirectIntegral}
\mathcal{H} \simeq \int^\oplus_{\sigma(\mathcal{N})} \mathcal{H}_\nu 
\end{equation}
and, under this isometry
\begin{equation}
\label{eq:SpectralN}
\mathcal{N} \simeq \int^\oplus_{\sigma(\mathcal{N})} \nu
\end{equation}
where $\nu$ is an abbreviation of $\nu \times \id_\nu$ acting on $\mathcal{H}_{\nu}$.  A Hilbert space bundle is called trivial if all the spaces $\mathcal{H}_\nu$ are isomorphic to a fixed space $\mathcal{H}^{(0)}$, so that $\mathcal{H}$ is isomorphic to the space of square-integrable $\mathcal{H}^{(0)}$-valued functions on $\sigma(\mathcal{N})$, i.e., $\mathcal{H} \simeq L^2(\sigma(\mathcal{N}), \mathcal{H}^{(0)}, \lambda)$. Any Hilbert space bundle can be decomposed into a countable sum of trivial bundles. Theorem~\ref{T:CLT2} can then be applied separately within each trivial bundle.
\end{remark}

\bibliography{NonDemolition}

\begin{thebibliography}{10}

\bibitem{Kraus}
K.~Kraus.
\newblock {\em States, effects and operations}.
\newblock Springer, 1983.

\bibitem{Holevo}
A.S. Holevo.
\newblock {\em Statistical structure of quantum theory}.
\newblock Springer, 2001.

\bibitem{Maassen}
H.~Maassen and B.~K{\"u}mmerer.
\newblock Purification of quantum trajectories.
\newblock {\em Lecture Notes-Monograph Series}, 48:252--261, 2006.

\bibitem{Benoist2016}
T.~Benoist, V.~Jaksic, Y.~Pautrat, and C.-A. Pillet.
\newblock On entropy production of repeated quantum measurements i. general
  theory.
\newblock {\em arXiv preprint arXiv:1607.00162}, 2016.

\bibitem{BFPP17}
T.~Benoist, M.~Fraas, Y.~Pautrat, and C.~Pellegrini.
\newblock Invariant measure for quantum trajectories.
\newblock {\em arXiv preprint arXiv:1703.10773}, 2017.

\bibitem{BFFS}
M.~Ballesteros, M.~Fraas, J.~Fr{\"o}hlich, and B.~Schubnel.
\newblock Indirect acquisition of information in quantum mechanics.
\newblock {\em Journal of Statistical Physics}, pages 1--35, 2015.

\bibitem{guerlin}
C.~Guerlin, J.~Bernu, S.~Deleglise, C.~Sayrin, S.~Gleyzes, S.~Kuhr, M.~Brune,
  J.M. Raimond, and S.~Haroche.
\newblock Progressive field-state collapse and quantum non-demolition photon
  counting.
\newblock {\em Nature}, 448(7156):889--893, 2007.

\bibitem{BaBe}
M.~Bauer and D.~Bernard.
\newblock Convergence of repeated quantum nondemolition measurements and
  wave-function collapse.
\newblock {\em Phys. Rev. A}, 84(4):044103, 2011.

\bibitem{BaBeBe1}
M.~Bauer, D.~Bernard, and T.~Benoist.
\newblock Iterated stochastic measurements.
\newblock {\em J. Phys. A Math Th}, 45(49):494020, 2012.

\bibitem{BaBeBe2}
M.~Bauer, T.~Benoist, and D.~Bernard.
\newblock Repeated quantum non-demolition measurements: convergence and
  continuous time limit.
\newblock {\em Ann. H. Poincar{\'e}}, 14(4):639--679, 2013.

\bibitem{Finetti}
B.~De~Finetti.
\newblock La pr{\'e}vision: ses lois logiques, ses sources subjectives.
\newblock {\em Ann. Inst. Henri Poincar{\'e}}, 7(1):1--68, 1937.

\bibitem{AliprantisBorder}
Charalambos~D. Aliprantis and Kim~C. Border.
\newblock {\em Infinite dimensional analysis}.
\newblock Springer, Berlin, third edition, 2006.
\newblock A hitchhiker's guide.

\bibitem{Reed-Simon-1}
Michael Reed and Barry Simon.
\newblock {\em Methods of modern mathematical physics. {I}}.
\newblock Academic Press, Inc. [Harcourt Brace Jovanovich, Publishers], New
  York, second edition, 1980.
\newblock Functional analysis.

\bibitem{Simon_Trace}
B.~Simon.
\newblock {\em Trace ideals and their applications}, volume 120 of {\em
  Mathematical Surveys and Monographs}.
\newblock American Mathematical Society, Providence, RI, second edition, 2005.

\bibitem{Ferguson}
T.~S. Ferguson.
\newblock {\em A course in large sample theory}.
\newblock Chapman \& Hall London, 1996.

\bibitem{Wald}
A.~Wald.
\newblock Note on the consistency of the maximum likelihood estimate.
\newblock {\em The Annals of Mathematical Statistics}, 20(4):595--601, 1949.

\bibitem{Kallenberg}
Olav Kallenberg.
\newblock {\em Foundations of modern probability}.
\newblock Probability and its Applications (New York). Springer-Verlag, New
  York, 1997.

\end{thebibliography}
\bibliographystyle{unsrt}

\end{document}